\documentclass[journal,twoside]{IEEEtran}

\usepackage{lipsum}
\usepackage{amsmath}
\usepackage{amsthm}
\usepackage{mdwmath}
\usepackage{amsfonts}
\usepackage{graphicx}
\usepackage[font=small]{caption}
\usepackage{epstopdf}
\usepackage[mathscr]{euscript}
\usepackage{mdwmath}
\usepackage{bbm}
\usepackage{mathrsfs}
\usepackage{bm}
\usepackage[tight,footnotesize]{subfigure}
\usepackage{comment}
\usepackage{marginnote}
\usepackage{color}
\usepackage{tabularx}
\usepackage[flushleft]{threeparttable}
\usepackage{multirow}
\usepackage{algorithm}
\usepackage[]{algpseudocode}
\usepackage{float}
\makeatletter
\newcommand{\algmargin}{\the\ALG@thistlm}
\makeatother
\newlength{\whilewidth}
\algdef{SE}[parWHILE]{parWhile}{EndparWhile}[1]
{\parbox[t]{\dimexpr\linewidth-\algmargin}{%
		\hangindent\whilewidth\strut\algorithmicwhile\ #1\ \algorithmicdo\strut}}{\algorithmicend\ \algorithmicwhile}%
\algnewcommand{\parState}[1]{\State%
	\parbox[t]{\dimexpr\linewidth-\algmargin}{\strut #1\strut}}
\algnewcommand{\parComment}[1]{\Comment%
	\parbox[t]{\dimexpr\linewidth-\algmargin}{\strut #1\strut}}
\theoremstyle{plain}
\newtheorem{theorem}{Theorem}[]

\newtheorem{lemma}[]{Lemma}

\newcommand{\mydef}{:=}
\newcommand{\Rc}{\mathbb{R}}
\newcommand{\Cc}{\mathbb{C}}
\newcommand{\Nb}{N_{bus}}
\newcommand{\Ng}{N_{gen}}

\newcommand{\V}{\mathbf{V}}
\newcommand{\scrv}[1]{\mathit{v}_{#1}}
\newcommand{\scrw}[1]{\mathit{w}_{#1}}

\newcommand{\Vv}[1]{V_{#1}}
\newcommand{\Ww}[1]{W_{#1}}
\newcommand{\cV}[1]{V^\star_{#1}}
\newcommand{\cS}[1]{S^\star_{#1}}
\newcommand{\Ss}[1]{S_{#1}}

\newcommand{\Ybus}[1]{Y_{#1}}
\newcommand{\Gg}[1]{G_{#1}} 
\newcommand{\Bb}[1]{B_{#1}} 
\newcommand{\Vdq}[1]{V_{#1}} 
\newcommand{\Pp}[1]{P_{#1}}
\newcommand{\Qq}[1]{Q_{#1}}
\newcommand{\U}{\mathbf{U}}
\newcommand{\M}[1]{\mathbf{M}_{#1}}
\newcommand{\Hh}[1]{\mathbf{H}_{#1}}

\begin{document}

\title{A Holomorphic Embedding Based Continuation Method for Identifying Multiple Power Flow Solutions}
\author{
\IEEEauthorblockN{Dan Wu,$^{\dagger}$ \IEEEmembership{Member, IEEE} and Bin Wang,$^{\ddagger}$ \IEEEmembership{Member, IEEE}}
\thanks{$\dagger$: Laboratory for Information and Decision Systems, Massachusetts Institute of Technology, Cambridge, MA, danwumit@mit.edu; $\ddagger$: Department of Electrical and Computer Engineering, Taxes A\&M University, College Station, TX, binwang@tamu.edu}%
}
\maketitle
\begin{abstract}
In this paper, we propose an efficient continuation method for locating multiple power flow solutions. We adopt the holomorphic embedding technique to represent solution curves as holomorphic functions in the complex plane. The holomorphicity, which provides global information of the curve at any regular point, enables large step sizes in the path-following procedure such that non-singular curve segments can be traversed with very few steps. When approaching singular points, we switch to the traditional predictor-corrector routine to pass through them and switch back afterward to the holomorphic embedding routine. We also propose a warm starter when switching to the predictor-corrector routine, i.e. a large initial step size based on the poles of the Pad\'{e} approximation of the derived holomorphic function, since these poles reveal the locations of singularities on the curve. Numerical analysis and experiments on many standard IEEE test cases are presented, along with the comparison to the full predictor-corrector routine, confirming the efficiency of the method.
\end{abstract}


\section{Introduction}
\label{sec:intro}

 The electric power grid is a critical energy infrastructure for power generation, transmission, and distribution in modern society. The inherent nonlinearity of power grid introduces a great challenge to analyze its dynamical behaviors when subject to disturbances, especially when penetrated with a large amount of intermittent renewable energies. Identifying the region of attraction about the operating condition, i.e. a stable equilibrium point (SEP) of the underlying dynamical system, can significantly improve the situational awareness and, therefore, will be of great importance to avoid blackouts. Characterizing this region requires the knowledge of a special type of unstable equilibrium point (UEP) which is called the \emph{type-1} UEP \cite{pai2012:energy,chiang2011:direct}. Determining them usually requires locating all nearby equilibria. In classical model \cite{sauer1998:power}, equilibria are the solutions to the \emph{power flow equations} \cite{klos1975:non,johnson1977:extraneous,tamura1980multiple}. 
	
A \emph{high-voltage solution} in the range of $[0.9,~1.1]$\footnote{A more restricted range may be assumed to be $[0.95,~1.05]$ for transmission systems.} p.u. represents a steady state under which the system can be well-operated. This solution is usually the SEP in transient stability analysis, while other solutions are UEPs. For tree structured networks, the high-voltage solution is unique \cite{chiang1990:existence}. However, it is possible that mesh networks can have multiple high-voltage solutions with either circulated flow \cite{korsak1972:question} or reversed power flow \cite{nguyen2014:appearance}. Although being avoided in the normal operations, circulated flow can happen during fault transients. Meanwhile, the reversal of power flow can become very common in future power grids as distributed energy resources (DER) keeps penetrating to distribution networks. To better characterize stability region and to examine other high-voltage operating points,
finding multiple power flow solutions plays a key role.

Nowadays a single high-voltage solution can be solved very efficiently. For example, systems with about 10,000 buses can be solved within a second \cite{geiri:pf}. However, the largest system that can be provably solved for all solutions is the 14 bus system \cite{mehta2016:numerical}. The efficiency of solving a single high-voltage solution comes from the knowledge of a good initial guess for local solvers to converge. But it is rather hard to acquire appropriate initial guesses for other solutions. If using random seeds, the complexity increases exponentially as the system size increases. Therefore, a systematic method is required to find these solutions.

Early attempts to find multiple power flow solutions dates back to $1970$s when  \cite{tavora1972:equilibrium} examined a 3-node system which admits $0$, $2$, $4$ or $6$ solutions. In 1989, \cite{salam1989:parallel,salam1989:parallel2} introduced the probability-one homotopy continuation method to find all the complex-valued solutions to the power flow problem. The homotopy continuation method requires estimating the total number of solutions to the power flow problem, which is still an ongoing research. In 1982, \cite{baillieul1982:geometric} sharpened the solution number bound from the classic Bezout's bound, $2^{2N_{\rm{bus}}-2}$, to a combinatorial bound, $C_{2N_{\rm{bus}} -2}^{N_{\rm{bus}}-1}$, where $N_{\rm{bus}}$ is the number of nodes in a power grid. Recently,  \cite{mehta2016:numerical} applied a polyhedral homotopy continuation method to completely solve the IEEE standard 14-bus system by the Bernstein-Khovanskii-Kushnirenko (BKK) bound which is sharper than the Bezout's bound. However, evaluating the BKK bound is very expensive. To further explore a simpler bound,  \cite{chen2018:network} introduced the adjacent polytope bound, which is sharper than the BKK bound and more computable. 

While progressive, the homotopy method usually ends up with a huge amount of complex-valued solutions which are fictitious power flow solutions. To only identify actual power flow solutions,  \cite{ma1993} introduced the idea of curve design which connects different real solutions by some 1-dimensional curves. Following these curves power flow solutions can be reached one by one\footnote{A very special type of test cases can be solved much more efficiently by some techniques from algebraic geometry. Interested readers are referred to \cite{coss2018:locating}. However, there is no such efficient algebraic geometry method for solving a general power flow case at present.}. Though efficient, \cite{molzahn2013:counterexample} provided a counter-example for \cite{ma1993}. To rectify their method, an elliptical formulation of the power flow problem is used in \cite{lesieutre2015:efficient} to restrict the curve design on high dimensional ellipses. It helps solve all the standard IEEE test cases which can be verified by the homotopy method in a reasonable time\footnote{Currently, there is no rigorous theoretical guarantee to show that the elliptical formulation can always connect all the real solutions. It is an ongoing research.}, including the counter-example in \cite{molzahn2013:counterexample}. The existence and construction of elliptical formulation were provided in \cite{wu2017:algebraic} and extended to the optimal power flow problem to find multiple local extrema for hard problems in \cite{wu2018:deterministic}. 

The curve tracing routine performed in \cite{lesieutre2015:efficient,wu2017:algebraic,wu2018:deterministic} is a traditional predictor-corrector algorithm which adopted a quadratic predictor \cite{schwetlick1987:higher}, Newton's method for corrector, and an adaptive step-length control \cite{den1981:steplength}. Many variations of the predictor-corrector algorithm exist, 
however, most of them depend only on the local information or previously solved points. To accelerate the curve tracing, in this paper we design a new hybrid algorithm called the \emph{holomorphic embedding based continuation} (HEBC) method to replace the traditional predictor-corrector algorithm during most of the curve tracing periods. It applies the holomorphic embedding technique to quickly pass through the non-singular curve segments by utilizing the global information of that curve, and uses a predictor-corrector routine to travel across singularities. 

The holomorphic embedding method (HEM) was introduced by Trias \cite{trias2012:1stHEM} in 2012 as a new power flow solver. The basic idea is to parameterize a polynomial system by an extra free variable and acquires the solution curve information by power series. Early attempts to use parameterization and power series for solving power flows started with \cite{sauer1981:explicit} and followed by \cite{xu1998:series,de2007:non}. Recently, HEM was extended to some applications with different modelings \cite{trias2016:dc,asu2017:sdbifur,utk2018:hemplf,utk2019:hemmotor}. Two features of HEM are particularly useful in our circumstance to improve searching efficiency. First, HEM can release us from local predictor-corrector scheme and provide with very long arc steps on the solution curve. This can largely reduce the burden of repeatedly solving linear systems in the corrector part. Moreover, the smallest real-valued pole of Pad\'{e} approximation can be used to design an appropriate step length when passing through singular point. It avoids overly large step sizes to improve numerical stability, and keeps step sizes progressive to maintain efficiency.

The contributions of this paper are summarized below.
\begin{enumerate}
	\item Showed an equivalent curve design for the elliptical formulation of the power flow problem;
	\item Proposed a hybrid numerical continuation method HEBC  for finding multiple power flow solutions;
	\item Proposed a warm starter to quickly initiate the predictor-corrector routine for passing through singularities; 
	\item Showed that HEBC outperforms the traditional predictor-corrector algorithm \cite{lesieutre2015:efficient} for all the tested cases;
	\item Computed solution sets\footnote{Solution sets will be available online soon.} for several large test cases which currently are intractable by homotopy continuation method or the similar.
\end{enumerate}


\section{Description of Power Flow Problem}
\label{sec:pf}
Throughout this paper we adopt the power flow formulation in rectangular coordinates.

\subsection{Power Flow Equations in Rectangular Coordinates}
\label{subsec:pf_rectangular}
Consider a connected power grid with $N_{\rm{bus}}$ nodes. Let the node voltage vector be
\begin{equation}
	\V \mydef \V_{\rm{d}} + j\V_{\rm{q}} \label{eq:complex_voltage}
\end{equation}
where $\V \in \Cc^{N_{\rm{bus}}}$; $\V_{\rm{d}} \in \Rc^{N_{\rm{bus}}}$ and $\V_{\rm{q}} \in \Rc^{N_{\rm{bus}}}$ are the real and imaginary parts of $\V$, respectively .

For the PQ bus we have
\begin{equation}
	\cV{k} \sum_{n=1}^{N_{\rm{bus}}} \Ybus{n,k} \Vv{n} = \cS{k} \label{eq:complex_pf}
\end{equation}%
where $\Vv{k}$ and $\Vv{n}$ are the corresponding entries of $\V$; $\Ybus{n,k}$ is the $(n,k)$-th entry of the bus admittance matrix $\mathbf{Y} \in \mathbb{C}^{N_{\rm{bus}} \times N_{\rm{bus}}}$; $\Ss{k} \in \mathbb{C}$ is the complex power load at bus $k$; superscript star $\star$ represents the conjugate operator.

Separating the real and imaginary parts of Equation \eqref{eq:complex_pf} gives the two equations about a PQ bus
\begin{subequations}
	\begin{align}
	 \Pp{k} &= \Vdq{\rm{d}, \it{k}} \sum_{n=1}^{N_{\rm{bus}}} \big(\Gg{n,k} \Vdq{\rm{d},\it{n}} - \Bb{n,k} \Vdq{\rm{q},\it{n}} \big) \nonumber
	 \\ &+ \Vdq{\rm{q},\it{k}} \sum_{n=1}^{N_{\rm{bus}}} \big(\Gg{n,k} \Vdq{\rm{q},\it{n}} + \Bb{n,k} \Vdq{\rm{d},\it{n}} \big) \label{eq:Pld_eqt}		
	\\ \Qq{k} &=  \Vdq{\rm{q},\it{k}} \sum_{n=1}^{N_{\rm{bus}}} \big(\Gg{n,k} \Vdq{\rm{d},\it{n}} - \Bb{n,k} \Vdq{\rm{q},\it{n}} \big) \nonumber
	\\ &- \Vdq{\rm{d},\it{k}} \sum_{n=1}^{N_{\rm{bus}}} \big(\Gg{n,k} \Vdq{\rm{q},\it{n}} + \Bb{n,k} \Vdq{\rm{d},\it{n}} \big) \label{eq:Qld_eqt}
	\end{align} \label{eq:PQ_bus}%
\end{subequations}
where $\Pp{k} \le 0$ and $\Qq{k} \le 0$\footnote{Usually a load absorbs reactive power, but it can possibly generate reactive power. In that case $\Qq{k} \ge 0$.} are the fixed active and reactive power loads at bus $k$; $\Gg{n,k}$ and $\Bb{n,k}$ are the $(n,k)$-th entries of the bus conductance matrix $\mathbf{G}$ and the bus susceptance matrix $\mathbf{B}$\footnote{$\mathbf{Y}=\mathbf{G}+j\mathbf{B}$}; $\Vdq{\rm{d} ,\it{k}}$, $\Vdq{\rm{d},\it{n}}$, $\Vdq{\rm{q},\it{k}}$ and $\Vdq{\rm{q},\it{n}}$ are the corresponding entries of $\V_{\rm{d}}$ and $\V_{\rm{q}}$, which are unknown variables that should be determined.

For the PV bus we have
\begin{subequations}
	\begin{align}
	\Pp{k} &= \Vdq{\rm{d},\it{k}} \sum_{n=1}^{N_{\rm{bus}}} \big(\Gg{n,k} \Vdq{\rm{d},\it{n}} - \Bb{n,k} \Vdq{\rm{q},\it{n}} \big) \nonumber
	\\ &+ \Vdq{\rm{q},\it{k}} \sum_{n=1}^{N_{\rm{bus}}} \big(\Gg{n,k} \Vdq{\rm{q},\it{n}} + \Bb{n,k} \Vdq{\rm{d},\it{n}} \big) \label{eq:Pgen_eqt}		
	\\ V_{\rm{m},\it{k}}^2 &= \Vdq{\rm{d},\it{k}}^2 + \Vdq{\rm{q},\it{k}}^2 \label{eq:V_eqt}
	\end{align} \label{eq:PV_bus}%
\end{subequations}
where $\Pp{k}$ is a fixed active power injection at bus $k$ which is usually positive but can be negative; $V_{\rm{m},\it{k}}$ is the fixed voltage magnitude at bus $k$.

For the slack bus with an angle reference we have
\begin{subequations}
	\begin{align}
	V_{\rm{m},\it{s}}^2 &= \Vdq{\rm{d},\it{s}}^2 + \Vdq{\rm{q},\it{s}}^2 \label{eq:Vslack_eqt}		
	\\ 0 &= \Vdq{\rm{q},\it{s}} \label{eq:angle_eqt}
	\end{align} \label{eq:slack_bus}%
\end{subequations}
where subscript $s$ is the slack bus number; $\Vv{\rm{m},\it{s}}$ is the slack bus voltage magnitude. 

One can further substitute \eqref{eq:angle_eqt} in \eqref{eq:Vslack_eqt}, \eqref{eq:PV_bus} and \eqref{eq:PQ_bus} to eliminate $\Vdq{\rm{q},\it{s}}$. Finally, \eqref{eq:PQ_bus}, \eqref{eq:PV_bus}, and \eqref{eq:slack_bus} together are the power flow equations we will investigate in this paper. Note that they are in quadratic form, thus can be written succinctly as
\begin{equation}
	PF(\U) \mydef \{ f_i(\U) = \U^T \M{i} \U - r_i,~ i=1,\cdots,2 N_{\rm{bus}} \} \label{eq:pf}
\end{equation}%
where $\U \mydef [\V_{\rm{d}}^T~~ \V_{\rm{q}}^T]^T$ is the unknown variable vector; $\M{i} \in \mathbb{SR}^{2 N_{\rm{bus}} \times 2 N_{\rm{bus}}}$ is a symmetric constant matrix for the quadratic part; $r_i \in \Rc$ is the constant scalar part.

\subsection{Equivalent Curve Design of Elliptical Formulation of Power Flow Equations}
As introduced in Section~\ref{sec:intro}, \cite{molzahn2013:counterexample} presented a counter-example that fails the proposed algorithm in \cite{ma1993} for finding all the real-valued power flow solutions. Then, \cite{lesieutre2015:efficient} introduced the concept of elliptical formulation of power flow equations which substantially changes the topology of path following curves and succeeded for that example. Later, \cite{wu2017:algebraic} showed the existence of elliptical formulation under mild conditions and constructed it in a systematical way. 

We start our discussion with a given invertible linear map $\mathscr{E} \in \Rc^{2 N_{\rm{bus}} \times 2 N_{\rm{bus}}}$ that sends Equation \eqref{eq:pf} to a set of high dimensional ellipses $EF(\U)$. The construction of $\mathscr{E}$ can be found in \cite{wu2017:algebraic,lesieutre2015:efficient}. Consider 
\begin{displaymath}
	\mathscr{E}:PF(\U) \to EF(\U)
\end{displaymath}
with 
\begin{displaymath}
	EF(\U) \mydef \{ g_i(\U) = \U^T \Hh{i} \U - \gamma_i,~i=1,\cdots,2 N_{\rm{bus}} \}
\end{displaymath}
where $\Hh{i} \in \mathbb{SR}^{2 N_{\rm{bus}} \times 2 N_{\rm{bus}}}$ and $\Hh{i} \succ 0$; $\gamma_i > 0$.

Let $\mathscr{Z}(h;x)$ be the operator that takes the projection of $\{(x,y) | h(x,y) = 0 \}$ onto $x$; define
\begin{displaymath}
	EF_{l-}(\U) \mydef EF(\U) - \{g_l(\U)\}
\end{displaymath}
\begin{displaymath}
EF_{l,\alpha}(\U,\alpha) \mydef EF_{l-}(\U) \cup \{g_l(\U)-\alpha,~\alpha \in \Rc \}.
\end{displaymath}

Since $EF(\U)$ defines a determined algebraic system, its algebraic set is generically 0-dimensional in $\Rc^{2 N_{\rm{bus}}}$. By removing one equation from $EF(\U)$, $EF_{l-}(\U)$ acquires one degree of freedom and defines a 1-dimensional algebraic set in $\Rc^{2 N_{\rm{bus}}}$. On the other hand, adding one extra degree of freedom to $EF(\U)$ makes the algebraic set of $EF_{l,\alpha}(\U,\alpha)$ 1-dimensional in $\Rc^{2 N_{\rm{bus}}+1}$.
The following Lemma~\ref{lemma:1} shows an equivalence between these two 1-dimensional algebraic sets.

\begin{lemma}\label{lemma:1}
	\begin{displaymath}
		\mathscr{Z}(EF_{l-};\U) = \mathscr{Z}(EF_{l,\alpha};\U)
	\end{displaymath}
\end{lemma}
The proof is trivial and omitted here.
Next, we state the equivalent curve design of elliptical formulation in Theorem~\ref{thm:1}.

\begin{theorem} \label{thm:1}
	\begin{displaymath}
	\mathscr{Z}(EF_{l-};\U) = \mathscr{Z}\big(\{PF(\U) - \alpha \mathscr{E}^{-1} \mathbf{e}_l\};\U\big)
	\end{displaymath}
	where $\mathbf{e}_l \in \Rc^{2 N_{\rm{bus}}}$ is a unit column vector with the $j$-th entry being 1.
\end{theorem}

\begin{proof}
	By definition, $EF_{l,\alpha}(\U,\alpha)$ can also be expressed as \{$EF(\U) - \alpha \mathbf{e}_l$\}. Then we have
	\begin{subequations}
		\begin{align}
	\mathscr{E}^{-1} \big( EF(\U) - \alpha \mathbf{e}_l \big) &= \mathscr{E}^{-1} \big(EF(\U)\big) - \alpha \mathscr{E}^{-1} \mathbf{e}_l \nonumber
	\\ &= PF(\U) - \alpha \mathscr{E}^{-1} \mathbf{e}_l. \nonumber
		\end{align}
	\end{subequations} 
	
	Since $\mathscr{E}$ is an invertible linear map, it is a homeomorphism. Hence, 
	\begin{displaymath}
		\mathscr{Z}(EF_{l,\alpha};\U) = \mathscr{Z}\big(\{PF(\U) - \alpha \mathscr{E}^{-1} \mathbf{e}_l\};\U\big).
	\end{displaymath}
	
	Finally, by Lemma~\ref{lemma:1} we conclude that
	\begin{displaymath}
	\mathscr{Z}(EF_{l-};\U) = \mathscr{Z}\big(\{PF(\U) - \alpha \mathscr{E}^{-1} \mathbf{e}_l\};\U\big).
	\end{displaymath}
\end{proof}


\section{Holomorphic Embedding Technique}
\label{sec:holo}
Theorem~\ref{thm:1} states that the 1-dimensional curves derived from the elliptical formulation $EF_{l-}$ can be acquired alternatively from a particular parameterized power flow problem $PF(\U) - \alpha \mathscr{E}^{-1} \mathbf{e}_l$. In Section~\ref{sec:pf} this $\alpha$ is restricted to a real-valued scalar to support one extra degree of freedom. If we allow $\alpha$ to be a complex number, the parameterized curve resides in the complex plane and becomes a 2-dimensional surface in the real space. If this complex-value parameterized curve happens to be governed by holomorphic functions, it is called the holomorphic embedding. The advantage of being holomorphic is that the global information of the embedded curve is determined and singularities on the curve can be predicted by analytic continuation techniques. 

\subsection{Holomorphic Embedding of Power Flow Equations}
\subsubsection{PQ Bus Embedding}
We start with the basic complex power balance equation for PQ bus in Equation \eqref{eq:complex_pf}. Note that $\Ss{k} = \Pp{k} + j \Qq{k}$ 
we define
\begin{subequations}
	\begin{align}
		\Pp{k}(\alpha) &\mydef (1+K_{\rm{p},\it{k}} \alpha) \Pp{k,0} \\
		\Qq{k}(\alpha) &\mydef (1+K_{\rm{q},\it{k}} \alpha) \Qq{k,0}
	\end{align}
\end{subequations}%
where $\alpha \in \Cc$; $K_{\rm{p},\it{k}}$ and $K_{\rm{q},\it{k}}$ are obtained from $\mathscr{E}^{-1} \mathbf{e}_l$ for some $l$; $\Pp{k,0}$ and $\Qq{k,0}$ are the fixed starting active and reactive power which admit a known solution. 

If we define a new variable $\Ww{k} \mydef \Vv{k}^{-1}$ for $\Vv{k} \neq 0$, and restrict parameterized $\Ww{k}(\alpha)$ to be reflective such that $\Ww{k}(\alpha) = \Ww{k}(\alpha^\star)$, then Equation~\eqref{eq:complex_pf} can be written as
\begin{subequations}
	\begin{align}
		\sum_{n=1}^{N_{\rm{bus}}} \Ybus{n,k} \Vv{n}(\alpha) &= \bigg( (1+K_{\rm{p},\it{k}} \alpha) \Pp{k,0}-  j (1  \nonumber
		\\ & + K_{\rm{q},\it{k}} \alpha) \Qq{k,0} \bigg) \Ww{k}^\star(\alpha^\star) \label{eq:para_complex_pf} \\
		\Vv{k}(\alpha) \Ww{k}(\alpha) &= 1 \label{eq:para_VW}
	\end{align} \label{eq:para_PQ_bus}%
\end{subequations}%
Note that on the right hand side of \eqref{eq:para_complex_pf} we use $\Ww{k}^\star(\alpha^\star)$ instead of $\Ww{k}^\star(\alpha)$ since they are equal by the reflective property\footnote{A more detailed discussion on the reflective requirement can be found in \cite{trias2015:reflection}.}. 

Since $\Vv{k}(\alpha)$ and $\Ww{k}(\alpha)$ are holomorphic \cite{trias2015:reflection}, we can use power series to represent them. 
Then, \eqref{eq:para_PQ_bus} can be re-written as
\begin{subequations}
	\begin{align}
		\sum_{n=1}^{N_{\rm{bus}}} \bigg( \Ybus{n,k} \sum_{i=0}^{\infty} \scrv{n,i} \alpha^i \bigg) &= \bigg( (1+K_{\rm{p},\it{k}} \alpha) \Pp{k,0} - j (1 \nonumber
		\\&+ K_{\rm{q},\it{k}} \alpha) \Qq{k,0} \bigg)\sum_{i=0}^{\infty} \scrw{n,i}^\star \alpha^i \label{eq:para_complex_pf_Taylor} \\
		\sum_{i=0}^{\infty} \scrv{n,i} \alpha^i \sum_{i=0}^{\infty} \scrw{n,i} \alpha^i &= 1 \label{eq:para_VW_Taylor}
	\end{align} \label{eq:para_PQ_bus_Taylor}%
\end{subequations}%
where $\scrv{n,i}$ and $\scrw{n,i}$ are the power series coefficients.

Matching up coefficients for every monomial of $\alpha$ in \eqref{eq:para_complex_pf_Taylor} and \eqref{eq:para_VW_Taylor} we can solve $(\scrv{k,1}, \scrv{k,2}, \cdots)$ and $(\scrw{k,1}, \scrw{k,2}, \cdots)$ recursively as long as $\scrv{k,0}$ and $\scrw{k,0}$ are provided.

\subsubsection{PV Bus Embedding}
Next, we consider the holomorphic embedding for PV bus equations. 
To retain holomorphicity, we need to bring back the reactive power balance equation \eqref{eq:Qld_eqt} to \eqref{eq:PV_bus} and consider reactive power input as a new variable. Again, by defining $\Ww{k} \mydef \Vv{k}^{-1}$ for $\Vv{k} \neq 0$ and restricting parameterized $\Ww{k}(\alpha)$ to be reflective we have the holomorphic embedded equations
\begin{subequations}
	\begin{align}
		\sum_{n=1}^{N_{\rm{bus}}} \Ybus{n,k} \Vv{n}(\alpha) &= \bigg( (1+K_{\rm{p},\it{k}} \alpha) \Pp{k,0} - j \Qq{k}(\alpha) \bigg) \Ww{k}^\star(\alpha^\star) \label{eq:para_PQ_PV} \\
		\Vv{k}(\alpha) \Vv{k}^\star(\alpha^\star) &= \Vv{k,\rm{m}}^2 + K_{\rm{v},\it{k}} \alpha \label{eq:para_V_PV} \\
		\Vv{k}(\alpha) \Ww{k}(\alpha) &= 1 \label{eq:para_VW_PV}
	\end{align} \label{eq:para_PV_bus}%
\end{subequations}%
where $\Vv{k,\rm{m}} \in \Rc$ is the fixed voltage magnitude at bus $k$; $K_{\rm{v},\it{k}}$ is obtained from the corresponding entry of $\mathscr{E}^{-1} \mathbf{e}_l$.  

By the holomorphic structure, we represent parameterized unknowns $\Vv{n}(\alpha)$, $\Ww{k}(\alpha)$, and $\Qq{k}(\alpha)$ through their power series. Then, \eqref{eq:para_PV_bus} are re-written as
\begin{subequations}
	\begin{align}
		\sum_{n=1}^{N_{\rm{bus}}} \bigg( \Ybus{n,k} \sum_{i=0}^{\infty} \scrv{n,i} \alpha^i \bigg) &= \bigg( (1+K_{\rm{p},\it{k}} \alpha) \Pp{k,0}  \nonumber
		\\ & - j \sum_{i=0}^{\infty} q_{k,i} \alpha^i \bigg) \sum_{i=0}^{\infty} \scrw{n,i}^\star \alpha^i \label{eq:para_PQ_PV_Taylor} \\
		\sum_{i=0}^{\infty} \scrv{n,i} \alpha^i \sum_{i=0}^{\infty} \scrv{n,i}^\star \alpha^i &= \Vv{k,\rm{m}}^2 + K_{\rm{v},\it{k}} \alpha \label{eq:para_V_PV_Taylor} \\
		\sum_{i=0}^{\infty} \scrv{n,i} \alpha^i \sum_{i=0}^{\infty} \scrw{n,i} \alpha^i &= 1 \label{eq:para_VW_PV_Taylor}
	\end{align} \label{eq:para_PV_bus_Taylor}%
\end{subequations}%
where $q_{k,i}$'s are the power series coefficients of $\Qq{k}(\alpha)$.

Matching up coefficients for every monomial of $\alpha$ in \eqref{eq:para_PQ_PV_Taylor}, \eqref{eq:para_V_PV_Taylor} and \eqref{eq:para_VW_PV_Taylor} we can solve $u_i$, $w_i$, and $q_i$ as well.

\subsubsection{Slack Bus Embedding}
Consider the slack bus voltage magnitude equation \eqref{eq:Vslack_eqt}. Its holomorphic embedded equation is
\begin{equation}
	\Vv{\rm{s}}(\alpha) \Vv{\rm{s}}^\star (\alpha^\star) = \Vv{\rm{s},m}^2+K_{\rm{s}} \alpha \label{eq:slack_embed}
\end{equation}%
where $\Vv{\rm{s},m}$ is the slack bus voltage magnitude, $K_{\rm{s}}$ is the corresponding entry from $\mathscr{E}^{-1} \mathbf{e}_l$.

Substituting the power series of $\Vv{\rm{s}}(\alpha)$ into Equation~\eqref{eq:slack_embed} and matching up each monomial of $\alpha$ we have
\begin{subequations}
	\begin{align}
	\scrv{\rm{s},\it{i}} &=-\bigg( \sum_{n=1}^{i-1} \scrv{\rm{s},\it{n}} \scrv{\rm{s},\it{i-n}} \bigg)/(2 \scrv{\rm{s},0})~~\text{for}~~i \ge 2
	\\ \scrv{\rm{s},1} &=K_{\rm{s}}/(2 \scrv{\rm{s},0}) \label{eq:slack_recursive}
	\end{align}
\end{subequations}

Combining the corresponding equations from the PQ bus, PV bus and slack bus equations we finally solve the power series coefficients for each degree-$i$. In practice, every degree requires solving a real-valued linear system (sparse) with its size $(4N_{\rm{bus}}+N_{\rm{gen}}-3) \times (4N_{\rm{bus}}+N_{\rm{gen}}-3)$ where $N_{\rm{gen}}$ is the number of PV nodes. As $i$ goes to infinity, the power series converges to the actual curve in the convergence range. To compromise accuracy and speed, we usually stop at a given maximum degree $i_{\rm{max}}$\footnote{\cite{trias2012:1stHEM} claims that degree $i$ will deplete double precision digits after $60$. How to choose an appropriate $i_{\rm{max}}$ is beyond the scope of this paper. We choose $i_{\rm{max}}=15$ in our numerical experiments by empirical experience considering speed and accuracy.}.

\subsection{Pad\'{e} Approximation}
The above subsection shows that each node voltage (as well as reactive power at PV bus) can be embedded as a holomorphic function, and demonstrates a recursive way to obtain the coefficients. In practice the holomorphic function can only be evaluated by a finite sequence of power series. Thus, the accuracy of the sequence deteriorates when approaching the singularities of the holomorphic function. To achieve a better convergence performance and to predict the location of singular point, we further compute the Pad\'{e} approximation. It approximates the holomorphic function by a rational function in which the numerator and denominator are polynomials. 
According to \cite{stahl1989:convergence,stahl1997:convergence}, the Pad\'{e} approximation has the maximum convergent domain if the degrees of its numerator and denominator have the minimum difference. It provides a criterion for determining the best degree(s) that should be chosen.

Consider an embedded voltage variable $\scrv{k}(\alpha)$ for some $k$. Suppose its first $N$ coefficients are known.
\begin{equation}
	\scrv{k}(\alpha) = \sum_{n=0}^{\infty} \scrv{k,n} \alpha^n \approx \sum_{n=0}^{N} \scrv{k,n} \alpha^n
\end{equation}
Let its Pad\'{e} approximation be
\begin{equation}
		\sum_{n=0}^{N} \scrv{k,n} \alpha^n 
		= \sum_{n=0}^{N_{\rm{n}}} u_{k,n} \alpha^n / \sum_{n=0}^{N_{\rm{d}}} l_{k,n} \alpha^n%
\end{equation}%
where we specify $N_{\rm{n}}+N_{\rm{d}}=N$, $N_{\rm{n}} \ge N_{\rm{d}}$, and $N_{\rm{n}}-N_{\rm{d}} \le 1$.

To reach a unique coefficient set, let $l_{k,0}=1$. Matching up the coefficients for each monomial we can solve $u_{k,n}$'s and $l_{k,n}$'s in a $(N+1) \times (N+1)$ complex-valued sparse linear system. If we compute the power series to the maximum degree $i_{\rm{max}}$, the system size in the real space is $2(i_{\rm{max}}+1) \times 2(i_{\rm{max}}+1)$.

Once the Pad\'{e} approximation has been calculated, we can move along the parameterized curve by evaluating Pad\'{e} approximated values until a power mismatch threshold\footnote{In our numerical experiments, this threshold is set at $10^{-3}$ p.u.} has been reached. We can also compute the real-valued zeros to the denominator function of Pad\'{e}. These zeros reveal the locations of singularities on the parameterized curve, which can further assist us designing appropriate arc length for passing through these singular points by the traditional predictor-corrector algorithm. Next section will discuss these designs in detail.

\section{Holomorphic Embedding Based Continuation Method}
\label{sec:alg}
The proposed HEBC method can be divided into an outer loop part and an inner loop part. The outer loop focuses on new solution updates and sequential curve designs; while the inner loop primarily follows the curve fed by the outer loop and returns the solution set found on that curve. 

\subsection{Outer Loop for Solution Search}
To make this article self-sustained, we briefly explain the search strategies in the outer loop and summarize it in Algorithm~\ref{alg:outer_pf_search}. Interested readers can refer to \cite{lesieutre2015:efficient}.

\begin{algorithm}[tbhp]
	\caption{Outer Loop for Locating Power Flow Solutions}
	\label{alg:outer_pf_search}
	\begin{algorithmic}[1]
		\State Solving for a power flow solution $x_1$.
		\State Generating elliptical mapping $\mathscr{E}$ by algorithms in \cite{lesieutre2015:efficient,wu2017:algebraic}.
		\State $S \gets x_1$ \Comment{Initialize solution set}
		\State  $N_{\rm{solu}} \gets |S|$ \Comment{Initialize number of solutions}
		\State  $k \gets 0$ \Comment{Initialize counting number}
		\While{$k \neq N_{\rm{solu}}$}
		\State $k \gets k+1$ \Comment{Update counting number}
		\State $x_0\gets x_k$ \Comment{Update starting solution}
		\For{$l=1,~2,~\cdots,~N_{\rm{eqn}}$}
		\State Compute $\mathscr{E}^{-1} \mathbf{e}_l$ \Comment{Equivalent curve design}
		\State Algorithm~\ref{alg:HEBC} \Comment{HEBC Algorithm}
		\State{Return $S_{\rm{new}}$} \Comment{Return newly found solutions}
		\If{$S_{new}$ is not in $S$}
			\State  $S \gets S \cup S_{\rm{new}}$ \Comment{Update the solution set}
			\State  $N_{\rm{solu}}\gets |S|$ \Comment{Update the number of solutions}
		\EndIf
		\EndFor
		\EndWhile
	\end{algorithmic}
\end{algorithm}

We start Algorithm~\ref{alg:outer_pf_search} with a known solution $x_1$ which can be solved by Newton's method or other techniques\footnote{This step relies on the past extensive research of solving a high voltage solution to the power flow problem. Many mature solvers are able to do this job for very large systems.}. After several initialization steps, designing the curve $\{PF-\alpha\mathscr{E}^{-1} \mathbf{e}_l\}$ which is equivalent to $\{EF_{l-}\}$ for $l$. Following the curve from $l=1$ to the last one by Algorithm~\ref{alg:HEBC} (which will be discussed shortly below) and collect new solutions. When finished tracing curves, assigning the starting point $x_0$ to a newly found solution, say, $x_2$, and repeating the procedure. The whole loop terminates upon every solution having been assigned to a starting point. 

Algorithm~\ref{alg:outer_pf_search} presents a procedure to follow each curve sequentially. However, the curve designs at the same starting solution are independent with each other, suggesting a parallel computing framework to simultaneously trace these curves. The parallel computing is not performed in this article, but can be done with ease and increase speed drastically.
%

\subsection{Inner Loop for Curve Tracing}
Instead of tracing a curve by the traditional predictor-corrector algorithm, we apply the holomorphic embedding technique to quickly pass through the regular curve segments. The predictor-corrector algorithm is only executed for traveling across singularities. It is switched back to the holomorphic embedding as soon as current steps leave a singular point.

Figure~\ref{fig:pade1} shows four holomorphic steps on a selected curve from a 5-bus case \cite{salam1989:parallel}. They reach the singular point very quickly. On the other hand, the blue curve in Figure~\ref{fig:pade1} was generated by the traditional predictor-corrector algorithm. It took dozens of steps to reach the same singularity. 

\begin{algorithm}[tbhp]
	\caption{Holomorphihc Embedding Based Continuation}
	\label{alg:HEBC}
	\begin{algorithmic}[1]
		\State{Input selected curve $\mathscr{E}^{-1}\mathbf{e}_l$.}
		\State{Initialize the $1st$ step.}
		\For{$k=1:M$}
			\For{$k_h=1:N_h$}
				\State{Initialize the $1st$ holomorphic step size $\delta_h$.}
				\State{Prepare parameters for holomorphic embedding.}
				\State{Compute power series of holomorphic embedding.}
				\State{Compute Pad\'{e} approximation.}
				\parState{Evaluate voltage values from Pad\'{e} and update $\alpha_{k_h+1}$.}
				\parState{Evaluate power mismatch $dP_{\rm{mis}}$ from computed voltages.}
				\While{minimum pole $p_{\rm{min}}$ is not determined}
					\State{Compute roots $\{ \zeta_i \}$ from Pad\'{e} denominator.}
					\parState{$p_{\rm{min}} \gets \zeta_{\rm{min}}$ if the minimum real root $\zeta_{\rm{min}}$ has correct sign.}
				\EndWhile
				\parState{Increase $\delta_h$ while $dP_{\rm{mis}}< dP_{\rm{max}}$ and $|\text{current point}|< |p_{\rm{min}}|$.}
				\parState{Decrease $\delta_h$ while $dP_{\rm{mis}}\ge dP_{\rm{max}}$ or $|\text{current point}|\ge |p_{\rm{min}}|$.}
				\parState{Correct current holomorphic predicted point by Newton's method.}
				\If{Correction succeeds}
					\State{Record current point.}
				\Else
					\parState{Delete current point and compute a starter for switching algorithm.}
					\State{Break.}
				\EndIf
				\If{$\alpha_{k_h+1}\alpha_{k_h}<0$}
					\State{Find a solution nearby.}
					\If{Fail to locate the solution}
						\parState{Delete current point and compute a cold starter for switching algorithm.}
						\State{Break.}
					\Else
						\State{Record solution to solution set $S_{\rm{new}}$.}
						\parState{Check completeness of the curve; jump out Algorithm~\ref{alg:HEBC} if completed.}
					\EndIf
				\EndIf
				\If{$|\alpha_{k_h+1}-\alpha_{k_h}|<d\alpha_{h,\rm{min}}$}
					\State{Compute a starter for switching algorithm.}
					\State{Break.}
				\EndIf		
			\EndFor
			\State{Execute predictor-corrector routine.}
		\EndFor
		\State{Return solution set $S_{\rm{new}}$}
	\end{algorithmic}
\end{algorithm}

\begin{figure}[tbhp]
	\centering
	\subfigure[Holomorphic Steps ]{\label{fig:pade1}\includegraphics[width=0.48\columnwidth]{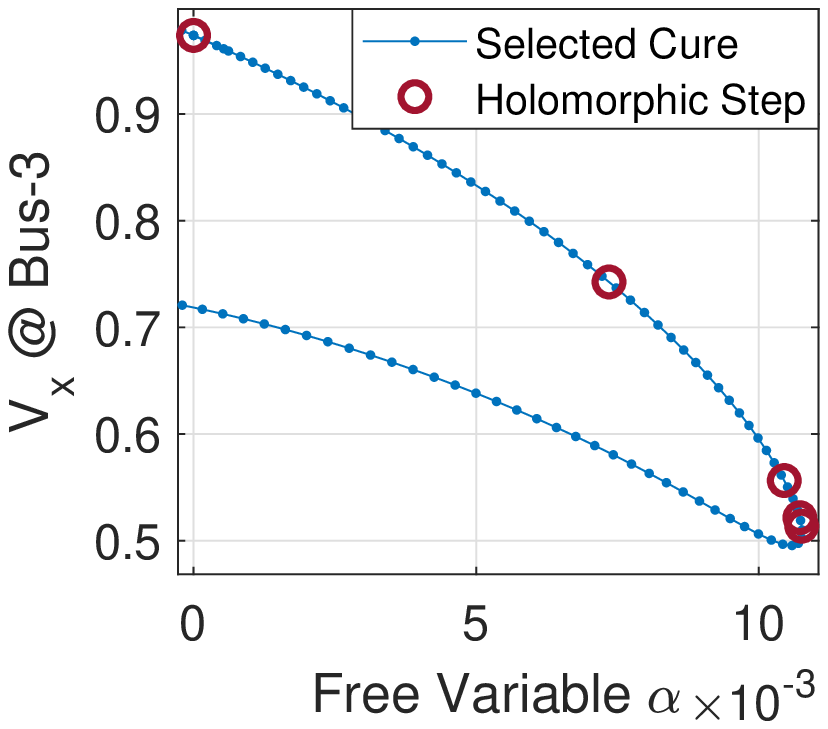}}~~
	\subfigure[Predictor-Corrector Steps ]{\label{fig:pade2}\includegraphics[width=0.48\columnwidth]{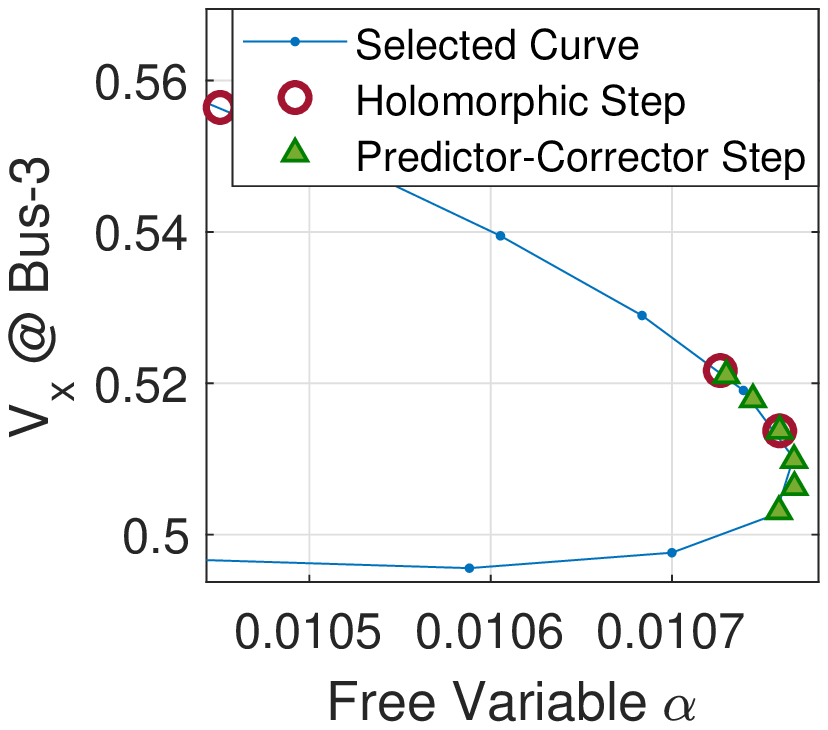}}
	\caption{Holomorphic Steps and Preditor-Corrector Steps} \label{fig:pade}
\end{figure}



\subsubsection{Criterion to Enter Predictor-Corrector Routine}
Two indicators are considered to trigger the switch of the algorithms in Algorithm~\ref{alg:HEBC}. The first indicator appears when the corrector steps fail to make the holomorphic prediction converge within a certain number of iterations. 
Another indicator comes when $|\alpha_{k_h+1}-\alpha_{k_h}|$ is smaller than a threshold value $d\alpha_{h,\rm{min}}$. Both suggest that current holomorphic step is close to singular (or at least badly scaled with respect to $\alpha$). 

\subsubsection{Using A Warm Starter to Accelerate Predictor-Corrector Steps}
One can initiate the predictor-corrector routine from a minimum step size, and increase it gradually. We refer it to a \emph{cold starter}. 
To avoid slow ``warming up'' steps, a warm starter is proposed and implemented. It relies on an estimated distance $d_{\rm{hp}}$ from the singular point to the last holomorphic point.
We specifically choose the initial step interval $S_{\rm{pc}}$ to be $1/5$ of the estimated distance $d_{\rm{hp}}$ and to be no greater than $0.45$ of the last holomorphic step size. Then, using $S_{\rm{pc}}$ to compute two backward steps to initiate a quadratic predictor. For example, the first two green triangles on the upper curve segment in Figure~\ref{fig:pade2} are the backward points evaluated by Pad\'{e} approximation at the step length $S_{\rm{pc}}$. It makes the predictor-corrector routine quickly pass through the singular point as shown by the rest green triangles.

\subsubsection{Criterion to Exit Predictor-Corrector Routine}
When travelling across a singular point, the direction of curve changes. Numerically, there exists a particular step $m_c$ such that
	$(\alpha_{m_c}-\alpha_{m_c-1})(\alpha_{m_c+1}-\alpha_{m_c})<0$.
After this moment, we continue the predictor-corrector routine for a while until the curve's slope value returns from infinity back to a tractable value. Instead of evaluating the actual slope of the curve, we monitor the maximum variable secant slope $R_{\rm{m}}$. 
\begin{equation}
	R_{\rm{m}} \mydef max\{ |(\Vv{k,m}-\Vv{k,m-1})/(\alpha_{m}-\alpha_{m-1})|,~\forall k \}
\end{equation}
As long as $R_{\rm{m}}$ drops to a threshold $R_{\rm{max}}$, say, $2 \times 10^4$, we jump out of the predictor-corrector routine and start a new sequence of holomorphic steps.

\section{Computational Complexity Comparison}
\label{sec:comp}

The holomorphic prediction consists of two sub-routines: 1) construct the power series; 2) compute Pad\'{e} approximation based on the power series. Both sub-routines require solving sparse linear systems. The sparsity reduces computational efforts in practice but makes analysis hard. To get a rough idea of the complexity, we simply assume the matrices are dense in the analysis, but solve them in sparse form practically.

In Section~\ref{sec:holo} computing the power series coefficients requires solving a sequence of linear systems up to the highest degree $i_{\rm{max}}$. A favorable observation is that all these linear systems share the same constant matrix. Thus, the LU factorization only needs to be performed once, while forward and backward substitutions need to be performed $i_{\rm{max}}$ times to generate coefficients for all degrees. Therefore, the computational complexity for the power series is 
\begin{equation}
	C_{\rm{Tl}} = \frac{2}{3}\big(4 N_{\rm{bus}} + N_{\rm{gen}} -3\big)^3 + 2 \big(4 N_{\rm{bus}} + N_{\rm{gen}} -3\big)^2 i_{\rm{max}} \label{eq:cplx_Taylor}
\end{equation}

In Pad\'{e} approximation, the complexity is
\begin{equation}
	C_{\rm{Pd}} = \bigg( \frac{2}{3}\big(2 i_{\rm{max}} +2\big)^3 + 2 \big(2 i_{\rm{max}} +2\big)^2 \bigg)  \big(2 N_{\rm{bus}} -1\big) \label{eq:cplx_Pade}
\end{equation}%
The total complexity of a holomorphic prediction is
$C_{\rm{Holo}} = C_{\rm{Tl}} + C_{\rm{Pd}}$.

On the other hand, in the traditional predictor-corrector algorithm the Newton's iterations in correctors are the most computational complex part. Again, suppose a dense Jacobian matrix (sparse in practice) the complexity of solving one Newton's iteration is
\begin{equation}
	C_{\rm{Newton}} = \frac{2}{3}\big(2 N_{\rm{bus}} -1)^3 + 2 \big(2 N_{\rm{bus}} -1)^2 \label{eq:cplx_newton}
\end{equation}

Suppose $i_{max}$ is fixed, $\Ng = 0.2 \Nb$\footnote{The number of PV buses usually occupies a small fraction of the total number of buses.}, and each corrector takes 3 Newton's iterations to converge for both the holomorphic step and the traditional predictor-corrector step, we have
\begin{equation}
	R = \lim_{N_{\rm{bus}} \to \infty} \frac{C_{\rm{Holo}}+3C_{\rm{Newton}}}{3C_{\rm{Newton}}} = 4.087%
\end{equation}
It suggests that one holomorphic step takes about four predictor-corrector steps computations asymptotically with the dense matrix LU factorization. So an average holomorphic step size which is greater than 4 times the average step size of the predictor-corrector algorithm can potentially reduce the computational time under the same assumptions. 

\section{Numerical Experiments}
\label{sec:num}
This section presents a comprehensive numerical evaluation of the proposed HEBC method on several standard power system test cases including ``case3TS", ``case3'', ``case4gs'',``case4BBc'', ``case4BB0'', ``case5Salam'', ``case6ww'', ``case7Salam'', ``case9'', ``case14'', ``case30'', ``case33bw'', ``case39'', ``case57''\footnote{Tap ratios are removed in this case to reduce the number of solutions.}, which can be found in the Matpower libary \cite{Zimmerman}, and ``case5loop'' \cite{molzahn2013:counterexample}. To avoid numerical instability and structurally unstable solutions, small resistance at $10^{-4}$ p.u. is added to lossless lines. The HEBC method and the full predictor-corrector method are coded in Matlab R2017b and executed on a PC with $2.8$GHz Intel i7-7700HQ CPU and $16$GB RAM.

\subsection{Comparison To Homotopy Continuation Method}
To demonstrate the superiority of computational efficiency in finding multiple power flow solutions, we begin with a comparison of the proposed HEBC method to the homotopy continuation method. The homotopy continuation is performed by the PHCpack \cite{PHCpack}. 

The HEBC method finds all the actual power flow solutions in this comparison as well as case14\footnote{No existing literature claims complete solution sets for larger IEEE test cases.}. Figure~\ref{fig:phc} shows execution time (in logarithmic scale) comparison between two methods. For test cases smaller than 5 buses, the PHCpack runs faster than the proposed HEBC method. However, for cases more than 5 buses, the HEBC outperforms the homotopy continuation method substantially. Considering the HEBC method is coded in Matlab and is not optimized to reach the most computational performance, the time reductions from HEBC are impressive. Test cases larger than 9 buses cannot be solved by PHCpack within 24 hours, thus are not considered in this comparison\footnote{A more recent progress in \cite{chen2018:network} successfully reduced the computational time of case14 to 5 minutes, however, the proposed HEBC is still much faster.}.

\begin{figure}[tbhp]
	\centering
	\includegraphics[width=0.95\columnwidth]{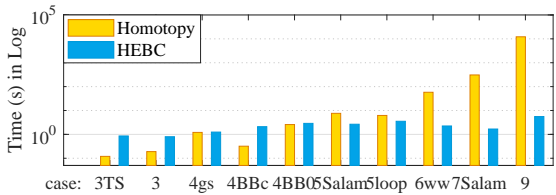}
	\caption{Comparison Between Homotopy Continuation and HEBC} \label{fig:phc}%
\end{figure}


\subsection{Comparison To Full Predictor-Corrector Algorithm}
In this part, we testify the traditional full predictor-corrector method from \cite{lesieutre2015:efficient} and the proposed HEBC method on the same set of test cases, and compare their numerical performances. Both methods provide the same solution sets for all cases, but the HEBC method is more efficient than the traditional predictor-corrector method.
Some hard\footnote{A curve is hard to follow in the sense that it contains too many singularities or some singular points are very sharp when turning directions.} sample curves are presented in Appendix~\ref{sec:append_1}. 
One can see from the left plots of Figure~\ref{fig:curves} that the traditional full predictor-corrector method, though with quadratic predictor and automatic step length adaption, takes very dense points to trace curves. On the other hand, the right plots of Figure~\ref{fig:curves} are primarily sparse. Small dense point periods only occur around singularities when HEBC switches to the predictor-corrector routine for passing through those singularities.
Summaries of the numerical results are collected in Table~\ref{table:1} and \ref{table:2}. 

\begin{table}[tbhp]
	{\footnotesize
		\caption{Numerical Results by Predictor-Corrector Method}\label{table:1}
		\begin{center}
			\begin{tabular}{l|c|c|c}
				\hline
				\multicolumn{1}{c|}{Method} & \multicolumn{2}{c|}{Predictor-Corrector} & \multicolumn{1}{l}{} \\ \hline
				\multicolumn{1}{l|}{Case}   & overall steps     & overall time (s)     & \# Solutions         \\ \hline
				3TS                     & 2962              & 1.257                & 6                    \\ \hline
				3                       & 3349              & 1.143                & 6                    \\ \hline
				4gs                     & 4281              & 1.405                & 6                    \\ \hline
				4BBc                    & 8838              & 2.660                & 12                   \\ \hline
				4BB0                    & 13791             & 3.719                & 14                   \\ \hline
				5Salam                  & 11465             & 3.626                & 10                   \\ \hline
				5loop                   & 17049             & 4.568                & 10                   \\ \hline
				6ww                     & 9421              & 3.209                & 6                    \\ \hline
				7Salam                  & 5195              & 1.978                & 4                    \\ \hline
				9                       & 22264             & 10.945               & 8                    \\ \hline
				14                      & 151423            & 102.401              & 30                   \\ \hline
				30                      & 5358518           & 6054.987             & 472                  \\ \hline
				33bw                    & 311957            & 249.736              & 16                   \\ \hline
				39                      & 3009935           & 3758.195             & 176                  \\ \hline
				57                      & 14647351          & 23864.005            & 606                  \\ \hline
			\end{tabular}
		\end{center}
	}
\end{table}

\begin{table}[tbhp]
	\begin{threeparttable}
	{\footnotesize
		\caption{Numerical Results by HEBC Method}	 \label{table:2}
			\begin{tabular}{l|c|c|c|c|c}
				\hline
				\multicolumn{1}{c|}{Method}  & \multicolumn{5}{c}{HEBC}                                                                                                                                                                                                                  \\ \hline
				\multicolumn{1}{c|}{Routine} & \multicolumn{2}{c|}{Holomorphic}         & \multicolumn{2}{c|}{Predictor-Corrector}  & \multirow{2}{*}{\begin{tabular}[c]{@{}c@{}}overall\\ time (s)\end{tabular}} \\ \cline{1-5}
				Case    & \# steps & time (s) & \# step            & time (s)            &                                                                             \\ \hline
				3TS                      & 253                           & 0.373    & 228                & 0.133               & 0.859                                                                       \\ \hline
				3                        & 290                           & 0.345    & 387                & 0.155               & 0.798                                                                       \\ \hline
				4gs                      & 481                           & 0.625    & 766                & 0.197               & 1.248                                                                       \\ \hline
				4BBc                     & 850                           & 1.174    & 1358               & 0.362               & 2.102                                                                       \\ \hline
				4BB0                     & 1212                          & 1.489    & 2805               & 0.773               & 2.890                                                                       \\ \hline
				5Salam                   & 1128                          & 1.754    & 1224               & 0.290               & 2.661                                                                       \\ \hline
				5loop                    & 1695                          & 2.448    & 1627               & 0.358               & 3.537                                                                       \\ \hline
				6ww                      & 995                           & 1.402    & 1143               & 0.288               & 2.249                                                                       \\ \hline
				7Salam                   & 564                           & 1.068    & 362                & 0.132               & 1.676                                                                       \\ \hline
				9                        & 1668                          & 3.158    & 4026               & 1.668               & 5.572                                                                       \\ \hline
				14                       & 13350                         & 34.443   & 27238              & 12.784              & 50.013                                                                      \\ \hline
				30                       & 403181                        & 2077.966 & 910664             & 674.828             & 2813.249                                                                    \\ \hline
				33bw                     & 15904                         & 81.896   & 65351              & 43.519              & 129.323                                                                     \\ \hline
				39                       & 184458                        & 1247.794 & 1044166            & 930.796             & 2204.543                                                                    \\ \hline
				57                       & 835550                        & 10565.59 & 3078609            & 3598.361            & 14304.691                                                                   \\ \hline
			\end{tabular}
		\begin{tablenotes}
			\small
			\item HEBC provides the same solution sets for all the cases as in Table I.
		\end{tablenotes}
	}
	\end{threeparttable}
\end{table}
\begin{figure}[tbhp]
	\centering
	\includegraphics[width=1\columnwidth]{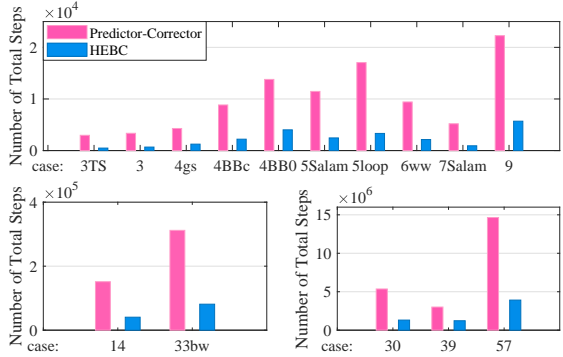} 
	\caption{Steps Needed for Different Cases with Full Predictor-Corrector and HEBC} \label{fig:num_step}
\end{figure}

Comparing the results in Figure~\ref{fig:num_step}, the total number of steps for HEBC is about $1/6$ to $1/3$ of the total number of steps for the full predictor-corrector method. This ratio, not surprisingly, should depend on the problem structure. In general, fewer singularities and longer horizontal curve segments favor the HEBC more. 

\begin{figure}[tbhp]
	\centering
	\includegraphics[width=0.95\columnwidth]{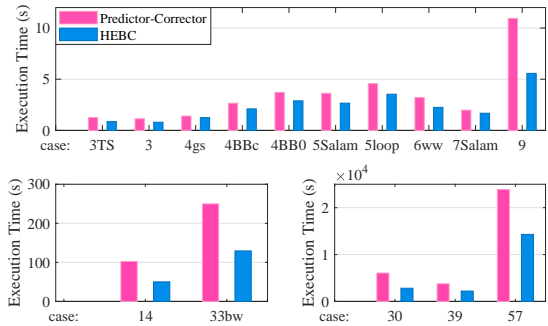} 
	\caption{Execution Times with Full Predictor-Corrector and HEBC}\label{fig:execution_time}
\end{figure}

To reveal the efficiency of HEBC, we compute the equivalent number of predictor-corrector steps $N_{\rm{eqv}}$
\begin{equation}
	N_{\rm{eqv}} \mydef (N_{\rm{pc}}-N_{\rm{he,pc}})/N_{\rm{he,holo}}
\end{equation}%
where $N_{\rm{pc}}$ is the number of full predictor-corrector steps; $N_{\rm{he,pc}}$ is the number of predictor-corrector routine steps in HEBC; and $N_{\rm{he,holo}}$ is the number of holomorphic routine steps in HEBC. 
From Table~\ref{table:1} and \ref{table:2} we calculate that one holomorphic step on average can represent $8.5$ predictor-corrector steps, with the worst case of $7$ steps and the best case of $15$ steps. 
In Figure~\ref{fig:execution_time} the first $9$ small cases up to case7Salam show a limited time saving by HEBC. However, starting at case9 the HEBC method outperforms the full predictor-corrector method by up to $50 \%$ of the execution time. Larger cases also exhibit at least $30 \%$ time saving in the lower plots of Figure~\ref{fig:execution_time}.

\subsection{Average Number of Steps on Each Dimension}
Recall that the HEBC method calls the Newton's method at each step to correct the predicted point.
These predicted points are sequentially determined over the curve tracing process. Thus, the HEBC method can be regarded as a systematic way to choose initial points for solving the power flow equations, where the number of initial points equals the number of steps in Table~\ref{table:2}, i.e. the sum of entries in the second and forth columns for each case.
From this point of view, one can assess the efficiency of HEBC by computing the average number of initial points (steps) allocated in each dimension
\begin{equation}
R_{eq} \mydef N^{1/d}%
\end{equation}%
where $N$ is the total number of initial points, $d$ is the dimensionality of the problem. $R_{eq}$ represents the number of points required in each single dimension such that the total number of initial points composed by their direct combinations achieves the same amount of initial points $N$ for the whole $d$-dimensional problem. Specifically for our problem, $R_{eq}$ is computed as
\begin{equation}
R_{eq} = (N_{he,pc}+N_{he,holo})^{1/(2\Nb -1)}
\end{equation}

\begin{figure}[tbhp]
	\centering
	\includegraphics[width=0.95\columnwidth]{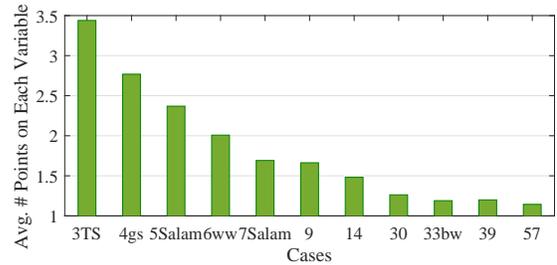} 
	\caption{Equivalent No. Random Seeds for Each Variable}\label{fig:avg_rand}
\end{figure}

Figure~\ref{fig:avg_rand} depicts the trend of $R_{eq}$ as system size increases. One can clearly see that the average number of steps distributed on each dimension decreases to nearly $1$. Hence, despite the increase of total number of steps, the average number of steps on each dimension seems to decrease in an asymptotic sense.

%
%


\section{Conclusions}
\label{sec:conclusions}
In this paper, we proposed an efficient hybrid method to solve multiple power flow solutions. We derived an equivalent curve design to the elliptical formulation of the power flow equations. Based on this design, a holomorphic embedding continuation method was introduced to replace the traditional predictor-corrector algorithm for regular curve tracing. Singular points were passed by the predictor-corrector routine. The complexity of one holomorphic step is around four times the complexity of a predictor-corrector step under certain assumptions. Numerical simulations showed that one holomorphic step size is equivalent to over eight predictor-corrector step size on average, and saved up to half of the computational time for some large test cases.

A possible future direction of research can use the proposed method to find multiple power flow solutions for dynamic stability analysis, especially in characterizing the stability boundary of an equilibrium point. Another interesting topic would be using this method for solving optimal power flow problems.

\bibliographystyle{ieeetr}
\bibliography{references}

\section*{Acknowledgments}
We would like to acknowledge the helpful discussions with Dr. Wenqiang Feng from DST Systems, Dr. Honghao Zheng from Siemens, Prof. Konstantin Turitsyn at MIT and Prof. Bernard Lesieutre at UW-Madison. The authors gratefully acknowledge support from the National Science Foundation under grant CRISP 1735513.

\appendices
\clearpage
\onecolumn
\newpage
\section{Sample Curves}
\label{sec:append_1}
Sample curves from simulations.

\begin{figure}[!ht]
	\centering
	\subfigure[Case14 by Predictor-Corrector]{\label{fig:14pc}\includegraphics[width=0.39\columnwidth]{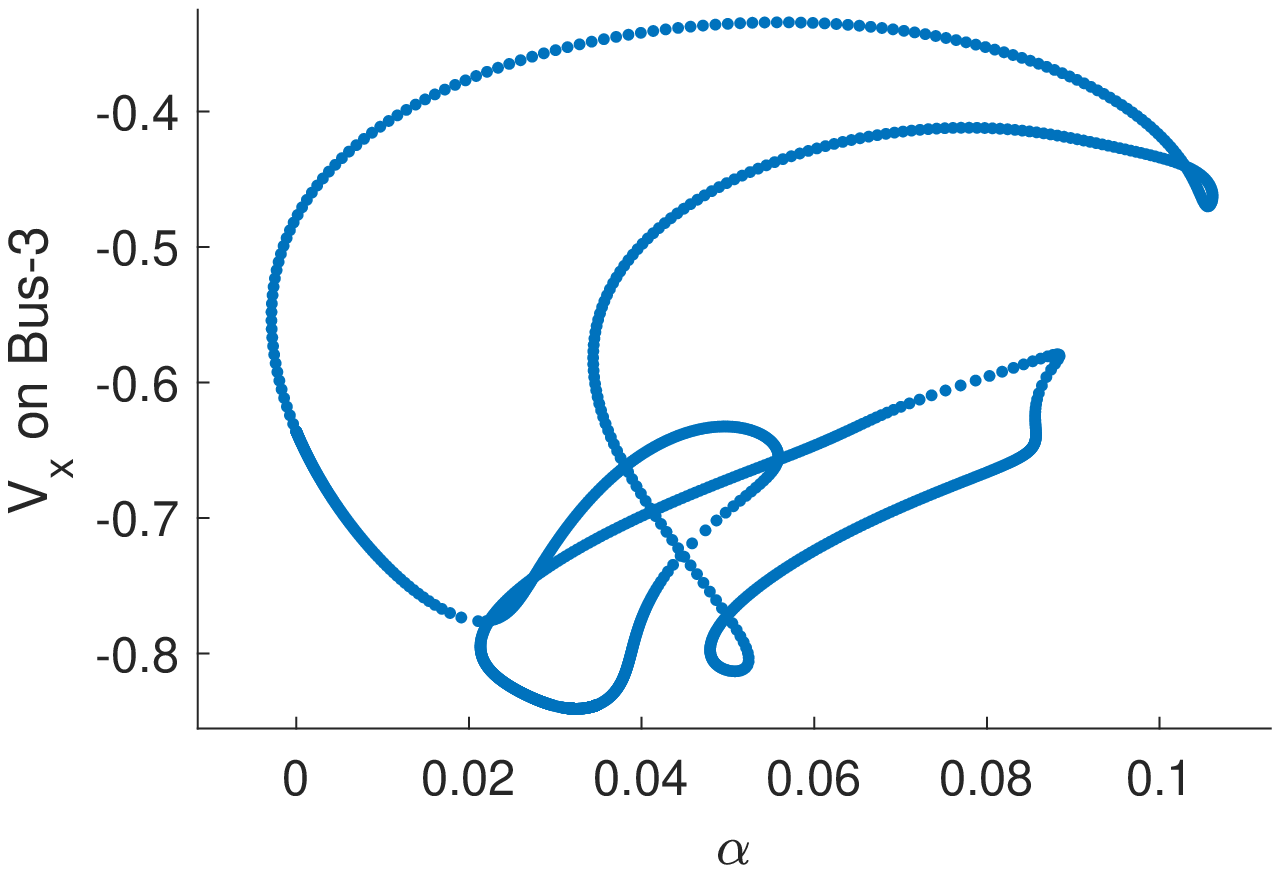}}~~~~
	\subfigure[Case14 by HEBC ]{\label{fig:14he}\includegraphics[width=0.39\columnwidth]{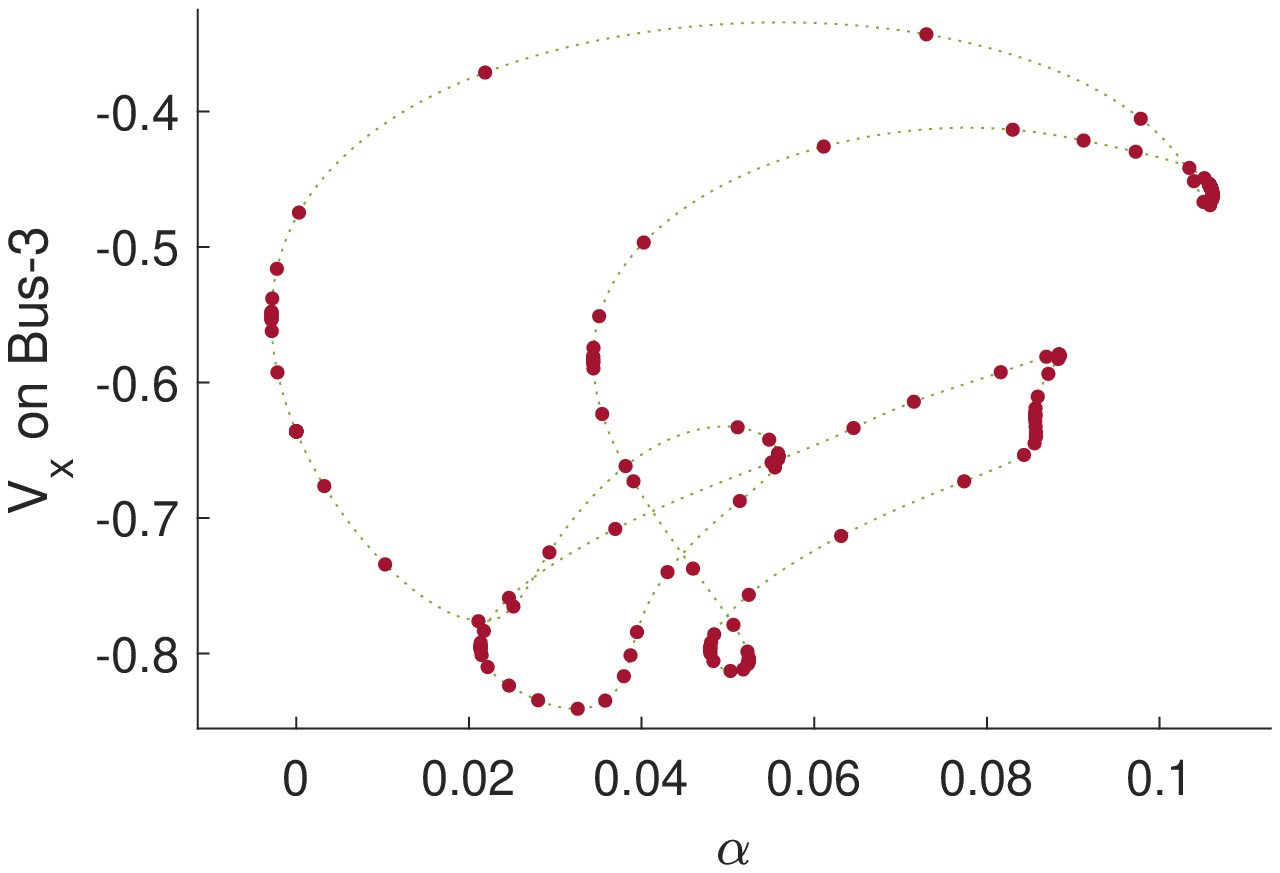}}\\
	\subfigure[Case30 by Predictor-Corrector]{\label{fig:30pc}\includegraphics[width=0.39\columnwidth]{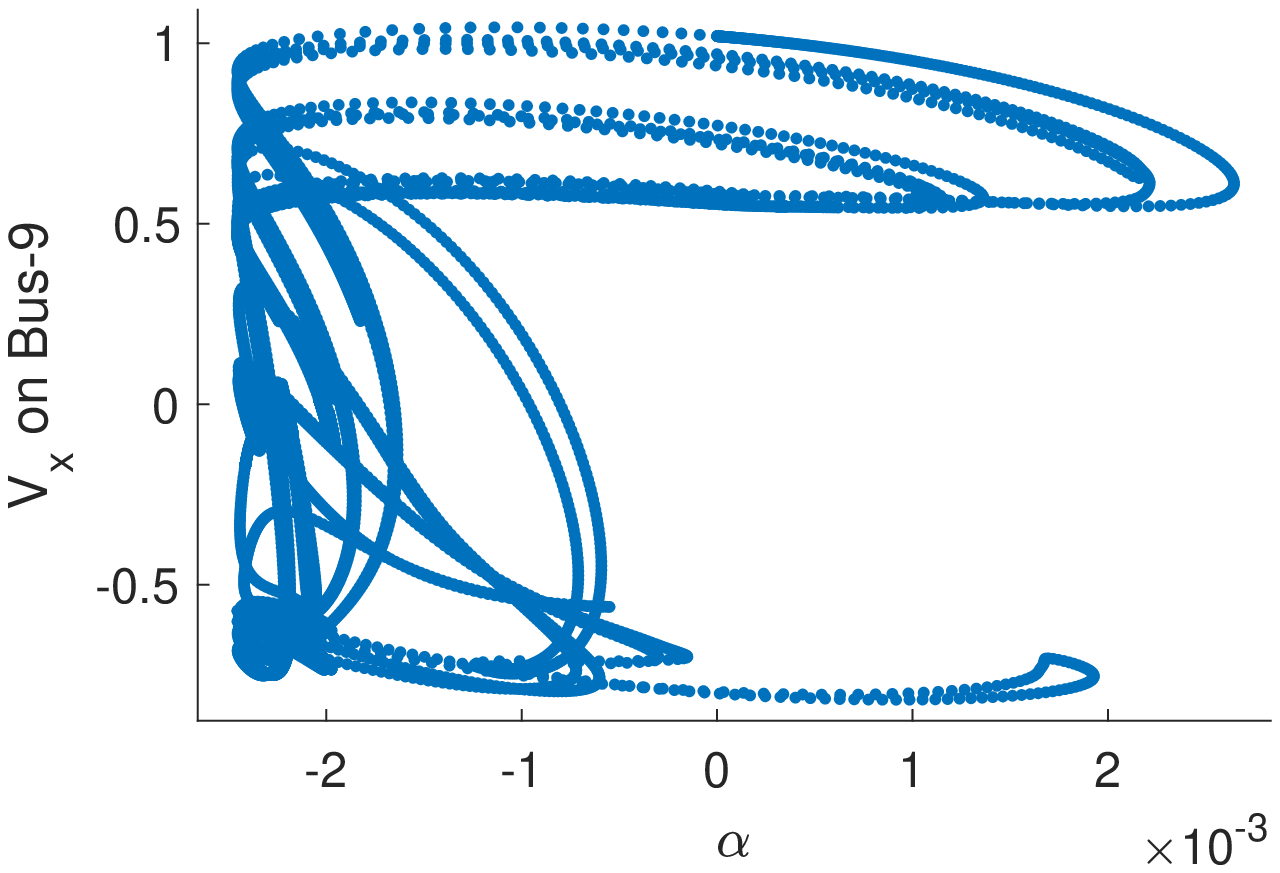}}~~~~
	\subfigure[Case30 by HEBC ]{\label{fig:30he}\includegraphics[width=0.39\columnwidth]{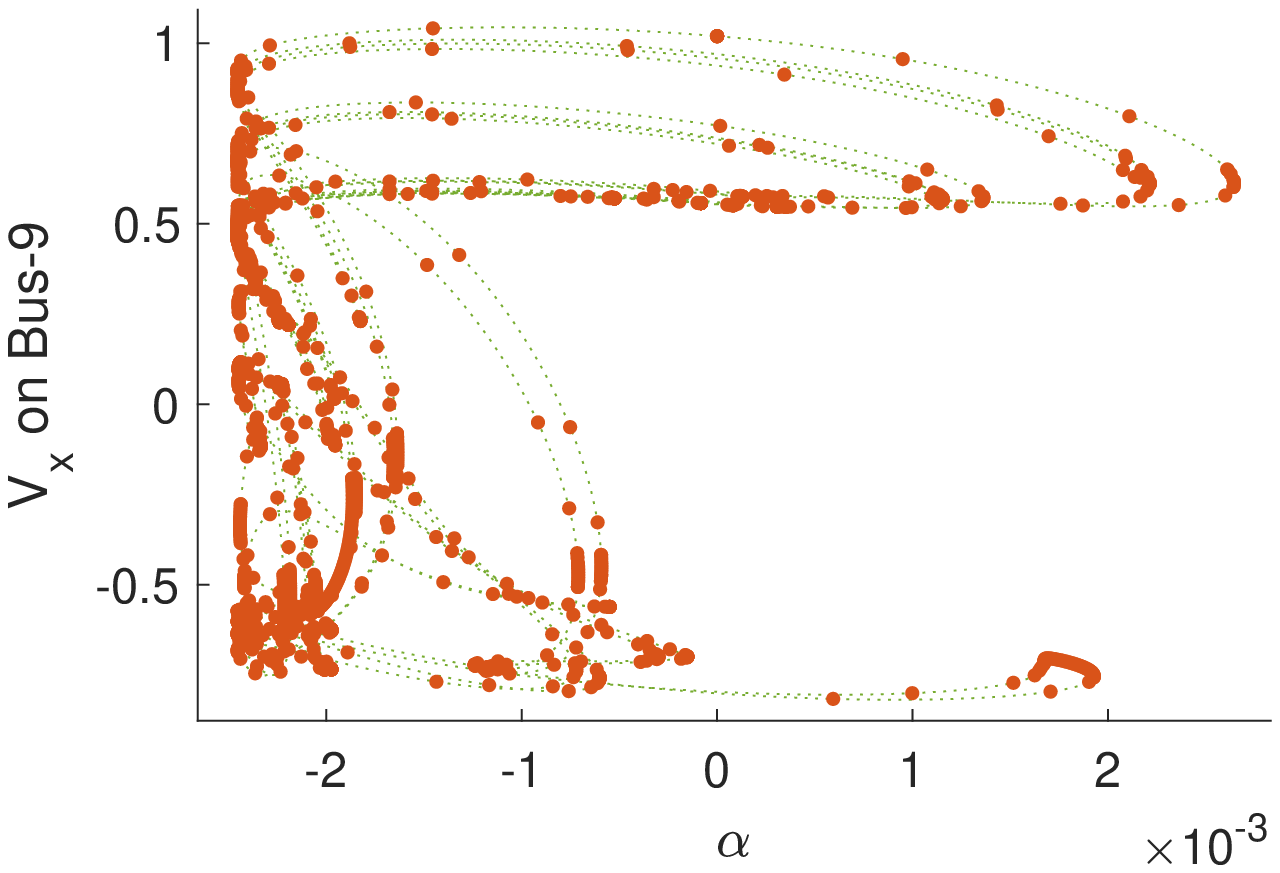}}\\
	\subfigure[Case33bw by Predictor-Corrector]{\label{fig:33pc}\includegraphics[width=0.39\columnwidth]{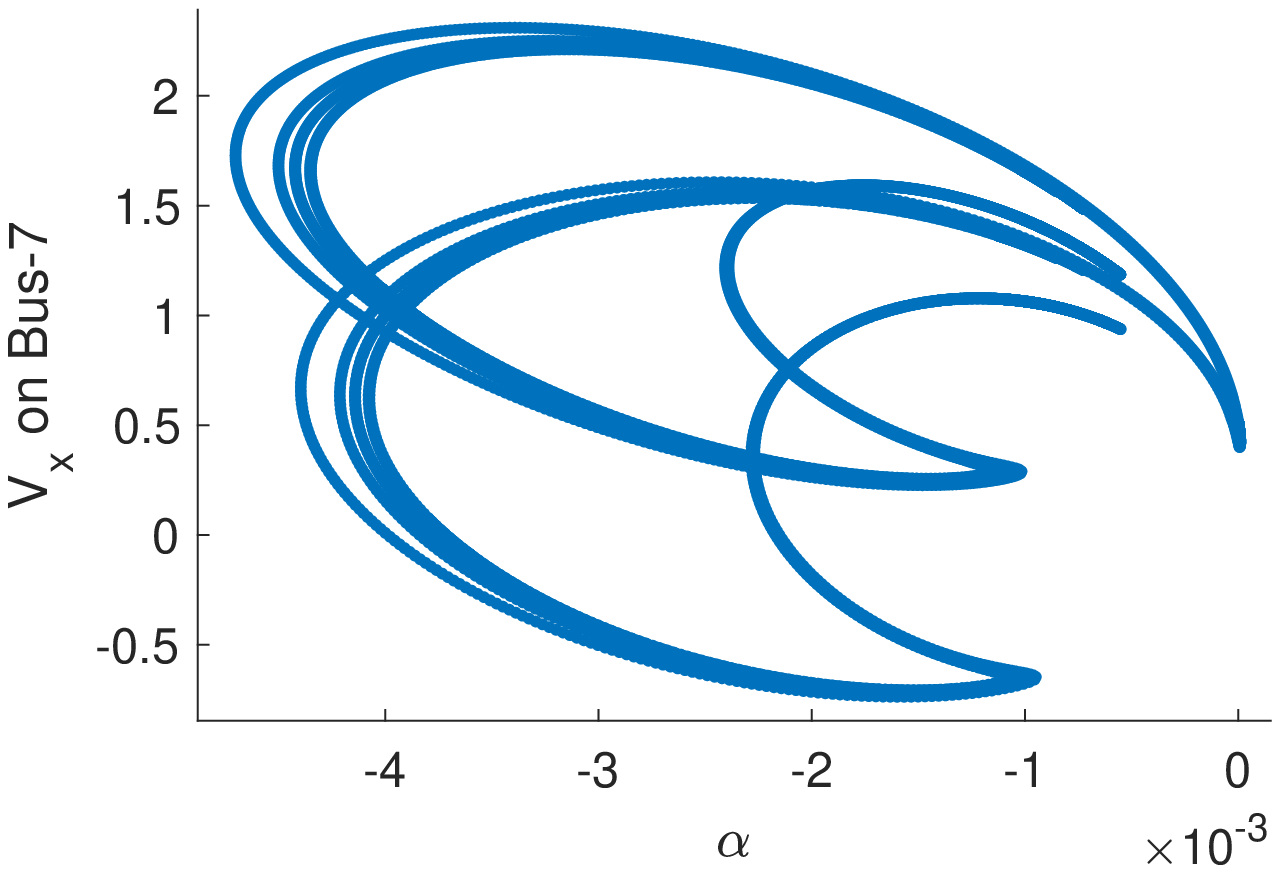}}~~~~
	\subfigure[Case33bw by HEBC ]{\label{fig:33he}\includegraphics[width=0.39\columnwidth]{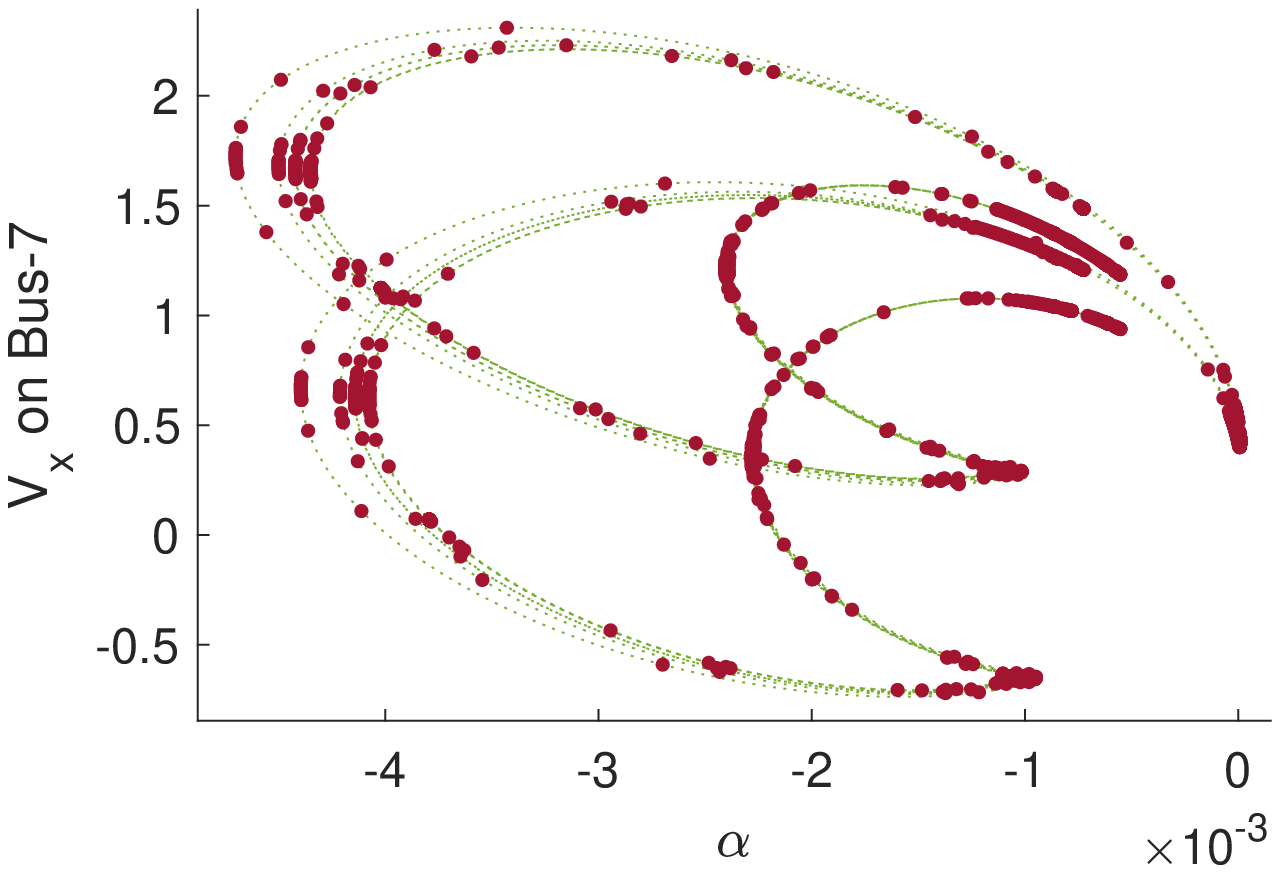}}\\
	\subfigure[Case57 by Predictor-Corrector]{\label{fig:57pc}\includegraphics[width=0.39\columnwidth]{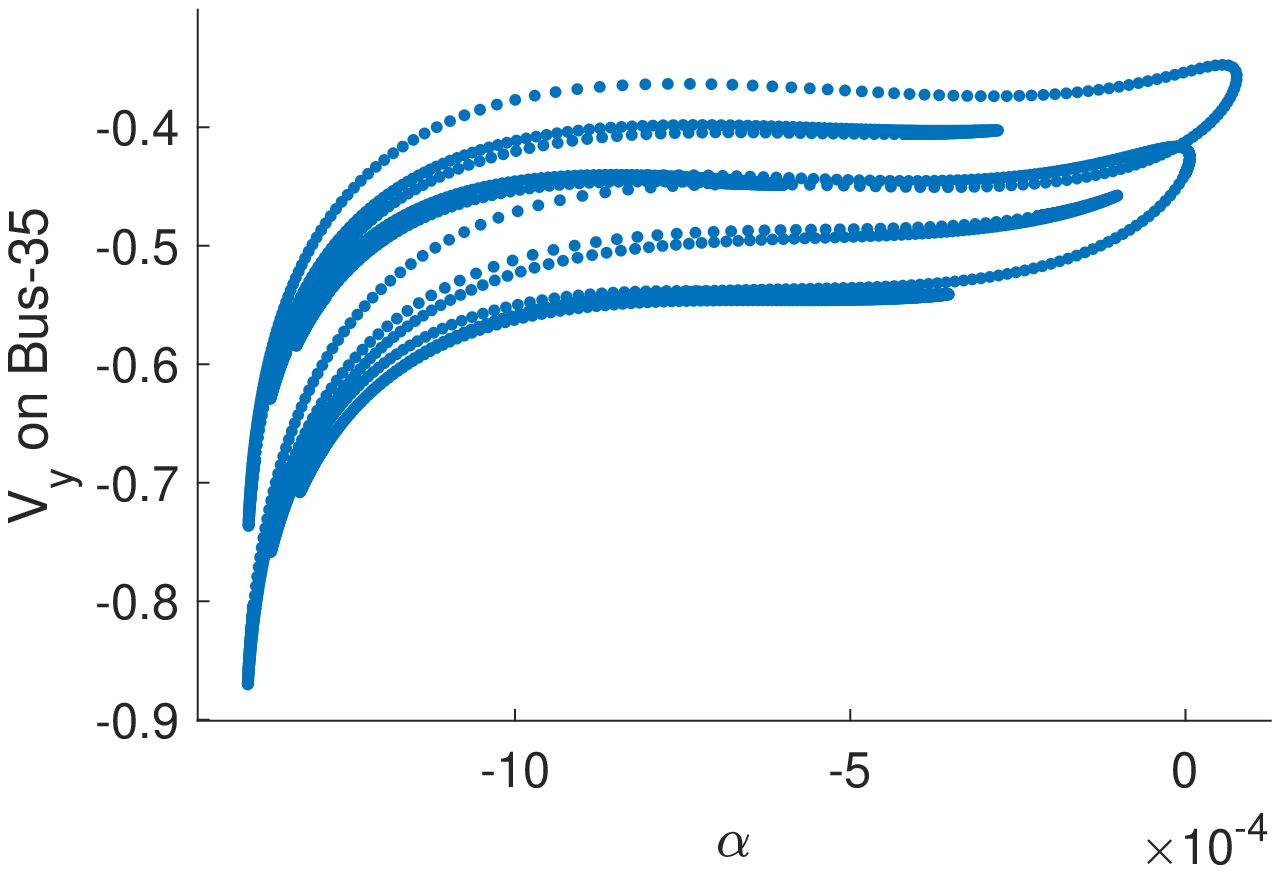}}~~~~
	\subfigure[Case57 by HEBC ]{\label{fig:57he}\includegraphics[width=0.39\columnwidth]{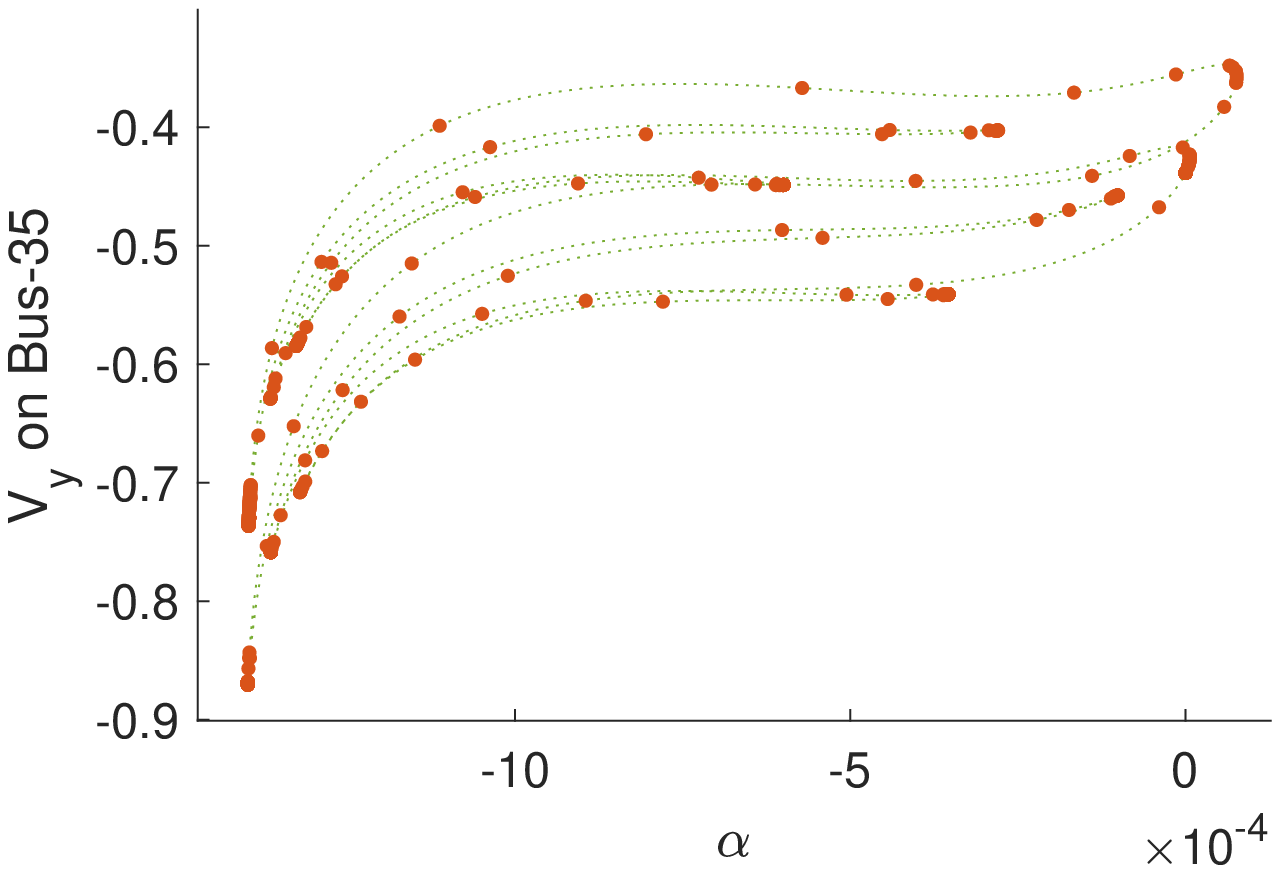}}
	\caption{Sample Curves Followed by Predictor-Corrector (Left) and HEBC (Right)} \label{fig:curves}
\end{figure}
%


\end{document}